\documentclass[a4paper,leqno]{amsart}

\usepackage{amsmath}
\usepackage{amsthm}
\usepackage{amssymb}
\usepackage{amsfonts}
\usepackage[applemac]{inputenc}
\usepackage{mathrsfs}
\usepackage{pdfsync}
\usepackage{color} 
\usepackage{mathtools}

\def\C{{\mathbb C}}
\def\R{{\mathbb R}}
\def\N{{\mathbb N}}
\def\Z{{\mathbb Z}}

\def\le{\leqslant}
\def\ge{\geqslant}

\newcommand{\eps}{\varepsilon}

\newcommand{\alp}{\boldsymbol{\alpha}}

\def\Eq#1#2{\mathop{\sim}\limits_{#1\rightarrow#2}}

\theoremstyle{plain}
\newtheorem{theorem}{Theorem}[section]
\newtheorem{lemma}[theorem]{Lemma}

\newtheorem{proposition}[theorem]{Proposition}
\newtheorem{hyp}{Assumption}

\theoremstyle{definition}
\newtheorem{definition}[theorem]{Definition}

\newtheorem{remark}[theorem]{Remark}
\newtheorem*{remark*}{Remark}

\numberwithin{equation}{section}

\begin{document}

\title[Nonlinear Dirac equation in honeycomb structures]{Rigorous derivation of nonlinear Dirac equations for wave propagation in honeycomb structures}

\author[J. Arbunich]{Jack Arbunich}

\author[C. Sparber]{Christof Sparber}

\address[J. Arbunich]
{Department of Mathematics, Statistics, and Computer Science, M/C 249, University of Illinois at Chicago, 851 S. Morgan Street, Chicago, IL 60607, USA}
\email{jarbun2@uic.edu}

\address[C.~Sparber]
{Department of Mathematics, Statistics, and Computer Science, M/C 249, University of Illinois at Chicago, 851 S. Morgan Street, Chicago, IL 60607, USA}
\email{sparber@uic.edu}

\begin{abstract}
We consider a nonlinear Schr\"odinger equation in two spatial dimensions subject to a periodic honeycomb lattice potential. Using a multi-scale expansion 
together with rigorous error estimates, we derive an effective model of nonlinear Dirac type. The latter describes the 
propagation of slowly modulated, weakly nonlinear waves spectrally localized near a Dirac point. 
\end{abstract}

\date{\today}


\thanks{This publication is based on work supported by the NSF through grant no. DMS-1348092}
\maketitle

\section{Introduction}\label{sec:intro}

Two-dimensional honeycomb lattice structures have attracted considerable interest in both physics and applied mathematics due to their unusual transport properties. 
This has been particularly stimulated by the recent fabrication of {\it graphene}, a mono-crystalline graphitic film in which electrons behave like two-dimensional {\it Dirac fermions} without mass, 
cf. \cite{CGPNG} for a recent review. Honeycomb structures also appear in nonlinear optics, modeling laser 
beam propagation in certain types of photonic crystals, see, e.g., \cite{BPS, HaMe} for more details.  
In both situations, the starting point is a two-dimensional Schr\"odinger operator 
\begin{equation}\label{H}
H = -\Delta + V_{\rm per}({\bf x}),\quad {\bf x}\in \R^2,
\end{equation}
where $V_{\rm per}\in C^\infty(\R^2;\R)$, denotes a smooth potential which is periodic with respect to a honeycomb lattice $\Lambda$ with fundamental cell $Y\subset \Lambda$ 
(see Section \ref{sec:pre} below for more details). C.~L. Fefferman and M.~I. Weinstein in their seminal paper \cite{FW1} proved that 
the associated quasi-particle dispersion relation generically exhibits conical singularities at the points of degeneracy, the so-called {\it Dirac points}. In turn, this yields effective equations of 
(massless) Dirac type for wave packets spectrally localized around these singularities, see \cite{AZ1, FW2, FW3}. 

To be more precise, let us recall that, by basic {\it Bloch-Floquet theory}, the spectrum 
$\sigma(H)\subset \R$ is given by a union of spectral bands, which can be obtained through the following ${\bf k}-$pseudo periodic boundary value problem \cite{Wi}:
\begin{equation}\label{EVP}
\left \{
\begin{aligned}
& H \Phi(\mathbf{y};\textbf{k}) = \mu(\textbf{k})\Phi(\mathbf{y};\textbf{k}), \ {\bf y} \in Y,\\
& \Phi(\mathbf{y} + \mathbf{v};\textbf{k}) =e^{i\textbf{k} \cdot \mathbf{v}}\Phi(\mathbf{y};\textbf{k}) , \ \mathbf{v} \in \Lambda ,
\end{aligned}
\right.
\end{equation}
where $\mathbf{k} \in Y^*$ denotes the wave-vector (or, quasi-momentum) varying within the Brioullin zone, i.e., the fundamental cell of the dual lattice $\Lambda^*$.  This yields 
a countable sequence of real-valued eigenvalues which are ordered, including multiplicity, such that
\[
\mu_{0}(\mathbf{k}) \le \mu_{1}(\mathbf{k}) \le \mu_{2}(\mathbf{k}) \le...,
\]
They describe the effective dispersion relation within the periodic structure. 
The corresponding pseudo-periodic eigenfunctions $\Phi_m(\cdot\, ;{\bf k})$, are known as {\it Bloch waves}. They form, for every fixed ${\bf k}\in Y^*$ 
a complete orthonormal basis on $L^2(Y)$. This consequently
allows one to write the linear time-evolution associated to \eqref{H} as
\[
\Psi(t,{\bf x})= e^{-i Ht} \Psi_0({\bf x}) = \sum_{m\ge 1} \int_{Y^*} e^{-i \mu_m({\bf k}) t}\langle \Phi_m(\cdot\, ;{\bf k}), \Psi_0 \rangle _{L^2(\R^2) }\Phi_m({\bf x};{\bf k})\, d{\bf k},
\]
for any initial data $\Psi_0\in L^2(\R^2)$. 

In the following, let the index $m\ge 1$ be fixed (to be suppressed henceforth) and assume that at ${\bf k}={\bf K_\ast}$, the $m$-th band $\mu_\ast\equiv \mu({\bf K_\ast})$, has a Dirac point, such that 
\[
\text{Nullspace}(H-\mu_\ast) = \text{span} \left\{\Phi_1({\bf x}), \Phi_2({\bf x})\right \},
\]
cf. Section \ref{sec:pre} below for a precise definition. 
Furthermore, let $0<\eps\ll 1$ be a small (dimensionless) parameter, and assume that initially $\Psi_0=\Psi_0^\eps$ is a wave packet, spectrally concentrated around this Dirac point, i.e., 
\[
\Psi^\eps_0({\bf x}) = \eps \alpha_{0,1}(\eps {\bf x}) \Phi_1({\bf x}) + \eps \alpha_{0,2}(\eps {\bf x}) \Phi_2({\bf x})
\]
where $\alpha_{0,1}, \alpha_{0,2}\in \mathcal S(\R^2;\C)$ are some slowly varying and rapidly decaying amplitudes. The overall factor $\eps$ is thereby introduced to ensure that $\| \Psi_0 \|_{L^2} =1$. 
It is proved in \cite{FW3} that the corresponding solution $\Psi^\eps (t, \cdot)$ of the linear Schr\"odinger equation satisfies,
\[
\Psi^\eps (t, {\bf x}) \Eq \eps 0 \eps e^{-i \mu_\ast t} \Big(\alpha_1(\eps t, \eps {\bf x}) \Phi_1({\bf x}) + \alpha_2(\eps t, \eps {\bf x}) \Phi_2({\bf x})\Big),
\]
provided $\alpha_{1,2}$ satisfy the following massless Dirac system:
\begin{equation}\label{lindirac}
\left \{
\begin{aligned}    &  \partial_{t} \alpha_{1} + \overline{\lambda}_{\#} \big(\partial_{x_{1}} + i\partial_{x_{2}}\big)\alpha_{2}  =0, \quad {\alpha_1}_{\mid t=0} = \alpha_{0,1},\\ 
    &  \partial_{t} \alpha_{2} + \lambda_{\#}  \big(\partial_{x_{1}} -i \partial_{x_{2}}\big)\alpha_{1} =0,\quad {\alpha_2}_{\mid t=0} = \alpha_{0,2},
   \end{aligned}
    \right.
\end{equation}
where $0\not = \lambda_{\#}\in \C$ is some constant depending on $V_{\rm per}$. This approximation is shown to hold up to small errors 
in $H^s(\R^2)$, over time-intervals of order $T\sim \mathcal O(\eps^{-2+\delta})$, for $\delta>0$. The Dirac system \eqref{lindirac} consequently describes the 
dynamics of the slowly varying amplitudes on large time scales. Of course, since the problem is linear, solutions of any size (w.r.t. $\eps$) will satisfy the same asymptotic behavior. 

In the present work, we aim to generalize this type of result to the case of {\it weakly nonlinear} wave packets. Formally, this problem was 
described in \cite{FW2}, but without providing any rigorous error estimates. Similarly, in \cite{AZ1, AZ2} the authors derive 
several nonlinear Dirac type models by a formal multi-scale expansions in the case of 
{\it deformed} and {\it shallow} honeycomb lattice structures, respectively. In the latter case, the size of the 
lattice potential serves as the small parameter in the expansion. In contrast to that, we shall not assume that the periodic potential 
$V_{\rm per}$ is small, but rather keep it of a fixed size of order $\mathcal O(1)$, i.e., independent of $\varepsilon$, in $L^\infty(\R^2)$.
However, as always in nonlinear problems, the size of the solution becomes important (depending on the power of the nonlinearity). 
To this end, we find it more convenient to put ourselves in a {\it macroscopic reference frame}, which means that we re-scale
\[
t\mapsto \tilde t= \eps t, \quad x\mapsto \tilde {\bf x} = \eps {\bf x}, \quad \tilde \Psi^\eps (\tilde t, \tilde {\bf x}) = \eps \Psi^\eps(\eps t, \eps {\bf x}),
\] 
with $\| \Psi_0^\eps (t,\cdot)\|_{L^2}=1$. We consequently consider the following, {\it semi-classically scaled} nonlinear 
Schr\"odinger equation (NLS), after having dropped all ``\, $\tilde { \ }$\, " again:
\begin{equation}\label{NLS}
i\varepsilon\partial_{t}\Psi^{\varepsilon}   = - \varepsilon^{2}\Delta \Psi^{\varepsilon}  + V_{\rm per}\Big(\frac{\mathbf{x}}{\varepsilon}\Big)\Psi^{\varepsilon} + 
\kappa^\varepsilon |\Psi^{\varepsilon}|^{2}\Psi^{\varepsilon}, \quad \Psi^\eps_{\mid t=0} = \Psi^\eps_0({\bf x}).
\end{equation}
Notice that in this reference frame the honeycomb lattice potential becomes {\it highly oscillatory}. We restrict ourselves to the case of NLS with cubic nonlinearities, since they are the most important ones 
within the context of nonlinear optics (a generalization to other power law nonlinearities is straightforward). However, 
we shall also remark on the case of {\it Hartree nonlinearities} in Section \ref{sec:hartree} below, as they are more
natural in the description of the mean-field dynamics of electrons in graphene \cite{HaMe, HLS}. The nonlinear coupling constant $\kappa^\varepsilon\in \R$, is assumed to be of the form
\[
\kappa^\varepsilon = \eps \kappa, \quad  \text{with $\kappa = \pm 1$.} 
\]
As will become clear, this size is {\it critical} with respect to our asymptotic expansion, in the sense that the associated modulating amplitudes $\alpha_{1,2}$ will 
satisfy a nonlinear analog of \eqref{lindirac}. For smaller $\kappa^\eps$, no nonlinear effects are present in leading order, as $\eps \to 0$, 
whereas for $\kappa^\eps$ larger than $\mathcal O(\eps)$, we do not expect the 
Dirac model to be valid any longer. Alternatively, this can be reformulated as saying that we consider asymptotically small solutions of critical size 
$\mathcal O(\sqrt{\eps})$ in $L^\infty$ to  
\eqref{NLS} but with fixed coupling strength $\kappa=\mathcal O(1)$. The advantage of our scaling is that it yields an asymptotic description for solutions of order $\mathcal O(1)$ in $L^\infty$ on macroscopic space-time scales, in contrast to the setting of \cite{AZ1, FW2}, in which the size of the solution varies as $\eps \to 0$. Another advantage is that this scaling 
puts us firmly in the regime of weakly nonlinear, semi-classical NLS with highly oscillatory periodic potentials, which have been extensively studied in, e.g., \cite{BLP, CS1, CS2, CMS, GMS}, albeit not 
for the case of honeycomb lattices. From the mathematical point of view, the present work will follow the ideas presented in \cite{GMS} and adapt them to the current situation.

Our main result, described in more detail in Section \ref{sec:stab}, rigorously shows
that solutions to \eqref{NLS} with initial data $\Psi_0^\eps$ spectrally localized around a Dirac point, 
can be approximated via
\begin{equation}\label{approx}
\Psi^\eps (t, {\bf x}) \Eq \eps 0 e^{-i \mu_\ast t/\eps} \left(\alpha_1( t, {\bf x}) \Phi_1\left(\frac{{\bf x}}{\eps}\right) + \alpha_2(t,{\bf x}) \Phi_2\left(\frac{{\bf x}}{\eps}\right) \right) +\mathcal O(\eps),
\end{equation}
where the amplitudes $\alpha_{1,2}$ solve the following nonlinear Dirac system
\begin{equation}\label{nldirac}
\left \{
\begin{aligned} 
& \partial_{t} \alpha_{1} + \overline{\lambda}_{\#} \big(\partial_{x_{1}} + i\partial_{x_{2}}\big)\alpha_{2} = -i\kappa\Big(b_{1} |\alpha_{1}|^{2} + 2b_{2}|\alpha_{2}|^{2}\Big)\alpha_{1} , \\ 
    &  \partial_{t} \alpha_{2} + \lambda_{\#}  \big(\partial_{x_{1}} -i \partial_{x_{2}}\big)\alpha_{1} = -i \kappa \Big( b_{1}|\alpha_{2}|^{2} + 2 b_{2}|\alpha_{1}|^{2}  \Big)\alpha_{2} ,
   \end{aligned}
    \right.
\end{equation}
subject to initial data $ \alpha_{0,1}({\bf x})$, $\alpha_{0,2}({\bf x})$, respectively, and with coefficients $b_{1,2}>0$ given by
\begin{equation*}
b_{1} = \int_{Y} |\Phi_{1}(\mathbf{y})|^{4} \;d\mathbf{y}=\int_{Y} |\Phi_{2}(\mathbf{y})|^{4} \;d\mathbf{y},\quad b_{2} =  \int_{Y} |\Phi_{1}(\mathbf{y})|^{2}|\Phi_{2}(\mathbf{y})|^{2} \;d\mathbf{y}.
\end{equation*}
The nonlinear Dirac system \eqref{nldirac} has been 
formally derived in \cite{FW2} and plays the same role as the coupled mode equations derived in \cite{GMS}, or the semi-classical transport equations derived in \cite{BLP, CMS}. 
This becomes even more apparent when we recall that the Bloch waves $\Phi_{1,2}$ can be written as
\[
\Phi_{1,2} \left(\frac{{\bf x}}{\eps}\right) =  \chi_{1,2}\left(\frac{{\bf x}}{\eps}\right)e^{i {\bf K_\ast}\cdot {\bf x}/\eps},
\]
where now $\chi_{1,2}(\cdot)$ is purely $\Lambda$-periodic. This shows that \eqref{approx} is of the form of a two-scale Wentzel-Kramers-Brillouin (WKB) ansatz, first introduced in \cite{BLP}, 
involving a highly oscillatory phase  function $$S(t,x)= {\bf K_\ast}\cdot {\bf x} - \mu_\ast t.$$ The latter is seen to be 
the unique, global in time, smooth solution (i.e., no caustics) of the {\it semi-classical Hamilton-Jacobi equation}
\[
\partial_t S + \mu (\nabla S) = 0, \quad S(0,{\bf x}) = {\bf K_\ast}\cdot {\bf x}.
\]
The fact that the phase function $S$ does not suffer from caustics, allows us to 
prove that our approximation \eqref{approx} holds for (finite) time-intervals of order $\mathcal O(1)$, bounded 
by the existence time of \eqref{nldirac}. In contrast to that, it is known from earlier results, cf. \cite{BLP, CMS}, that caustics within 
$S$ lead to the breakdown of a single-phase WKB approximation. Also note that the group velocity $\nabla S={\bf K_\ast}$ being constant  
allows us to localize around Dirac points. 

In terms of the unscaled variables used in \cite{FW2, FW3}, our approximation result is seen to hold on times of order $\mathcal O(\eps^{-1})$, which is considerably shorter than the time-scale $O(\eps^{-2+\delta})$ obtained in \cite{FW2}. 
This drawback, however, is expected due to the influence of our nonlinear perturbation and consistent with earlier results for 
linear semi-classical Schr\"odinger equations with periodic potentials and additional, slowly varying non-periodic 
perturbations $U=U(t,x)$, see, e.g., \cite{BLP}. In the case of the latter, one also expects the appearance of purely geometric effects, such as the celebrated 
Berry phase term (cf. \cite{CS1, CMS}). It would certainly be interesting to understand the corresponding Dirac-type dynamics already on the linear level, when such 
geometric effects are present (but this is beyond the scope of the current work).  
We finally note that other examples of coupled mode equations have been derived in, e.g., \cite{DoHe, DoUe, GHW, PS1, PS2}. In these models, however, the 
resulting mode equations are of transport type, and any coupling between the amplitudes stems purely from the nonlinearity, in contrast to the Dirac model. We finally remark that 
discrete mode equations, valid in the tight binding regime, have recently been studied in \cite{AZ2, FLW}.
\medskip

This paper is now organized as follows: In Section \ref{sec:pre} we recall some basic properties of honeycomb lattice structures and the associated lattice potentials. 
Then, we shall perform a formal multi-scale expansion in Section \ref{sec:multi}, for which we will set up a rigorous framework in Section \ref{sec:math}. 
As a last step, we shall prove in Section \ref{sec:stab} 
that the approximate solutions thereby obtained is stable 
under the nonlinear time evolution of the NLS. Finally, we shall briefly discuss the case of Hartree nonlinearities in Section \ref{sec:hartree}.

\section{Mathematical preliminaries}\label{sec:pre}

\subsection{Honeycomb lattice potentials}

We start by recalling the basic geometry of triangular lattices arising naturally in the connection to honeycomb structures. Let $ \Lambda = \mathbb{Z}\textbf{v}_{1} \oplus \mathbb{Z}\textbf{v}_{2}$, 
spanned by
\[
\mathbf{v}_{1} =a \begin{pmatrix} \frac{\sqrt{3}}{2} \\\\ \frac{1}{2}\end{pmatrix}, \
\mathbf{v}_{2} = a\begin{pmatrix}\frac{\sqrt{3}}{2}\\\\ -\frac{1}{2} \end{pmatrix}, \quad a>0.
\]
Then, the honeycomb structure $H$ is the union of two sub-lattices $\Lambda_{\bf A}= {\bf A} + \Lambda$ and $\Lambda_{\bf B} = {\bf B} + \Lambda$, cf. \cite{FW1, FW2}. 
The corresponding dual lattice $ \Lambda^{*} = \mathbb{Z}\mathbf{k}_{1} \oplus \mathbb{Z}\mathbf{k}_{2}$ is spanned by the dual basis vectors
\[
\mathbf{k}_{1} = q\begin{pmatrix}\frac{1}{2}\\\\ \frac{\sqrt{3}}{2}\end{pmatrix}, \ \mathbf{k}_{2} = q\begin{pmatrix}\frac{1}{2}\\\\ -\frac{\sqrt{3}}{2}\end{pmatrix}, \ q \equiv \frac{4 \pi}{a\sqrt{3}},
\]
such that $\mathbf{v}_{j} \cdot \mathbf{k}_{j} =2 \pi $. The fundamental period cell is denoted by 
\[
Y = \{ \theta_{1}v_{1} + \theta_{2}v_{2}: 0\le \theta_{j} \le 1, j=1,2 \}.
\]
In the following, we shall denote $L^{2}_{\rm per}\equiv L^{2}(\mathbb{R}^{2}/\Lambda)$ and likewise for other function spaces. An element 
$f \in L^{2}_{\rm per}$ is called $\Lambda$-periodic, since it satisfies
\[ f(\mathbf{y}+\mathbf{v})=f(\mathbf{y}), \quad \text{for all ${\bf y}\in \R^2$ and $\mathbf{v}\in \Lambda$.}\]
On the other hand, we shall denote $f \in L^{2}_{\mathbf{k}} $ and call it ${\bf k}-$pseudo periodic, if
\begin{align*}
f(\mathbf{y}+\mathbf{v})=f(\mathbf{y})e^{i\mathbf{k}\cdot \mathbf{v}}\quad \text{for all ${\bf y}\in \R^2$ and $\mathbf{v}\in \Lambda$.}
\end{align*}

The Brillouin zone, $Y^*$, is a choice of fundamental unit cell in the dual period lattice which we choose to be a regular hexagon centered at the origin. 
Due to symmetry, the vertices of $Y^*$ fall into two equivalence classes of points, 
$\mathbf{K}\equiv \frac{1}{3}(\mathbf{k}_{1}+\mathbf{k}_{2})$ and  $\mathbf{K}^{\prime} \equiv - \mathbf{K} = \frac{1}{3}(\mathbf{k}_{2} - \mathbf{k}_{1})$. 
The other vertices are generated by the action of $2 \pi/3$ rotation matrix, $R$ given by
\begin{equation*}
R =  \begin{pmatrix} -\frac{1}{2}&\frac{\sqrt{3}}{2}\\\\ -\frac{\sqrt{3}}{2}&-\frac{1}{2} \end{pmatrix},
\end{equation*}
so that
\begin{align*}
&R\mathbf{K} = \mathbf{K} +\mathbf{k}_{2},\quad  R^{2}\mathbf{K} = \mathbf{K} -\mathbf{k}_{1}. \\
&R\mathbf{K^{\prime}} = \mathbf{K^{\prime}} - \mathbf{k}_{2},\quad R^{2}\mathbf{K^{\prime}} = \mathbf{K^{\prime}} + \mathbf{k}_{1}.
\end{align*}
Let $\mathbf{K}_{\ast}$ be a vertex of $\mathbf{K}$ or $\mathbf{K}^{\prime}$ type. For any function $f\in  L^{2}_{\mathbf{K}_{\ast}}$, we introduce the unitary operator
\begin{equation}\label{R}
\mathcal R: f\mapsto \mathcal R[f]= f(R^{*}\mathbf{y})= f(R^{-1}\mathbf{y}).
\end{equation}
One checks that $\mathcal R$ has eigenvalues ${1,\tau,\overline{\tau}}$, where $\tau = e^{2\pi i/3}$. This consequently yields 
the following pairwise orthogonal subspaces
\begin{align*}
    &L^{2}_{\mathbf{K}_{\ast},1} \equiv \{ f \in L^{2}_{\mathbf{K}_{\ast}} : \mathcal R[f] = f \} \\
    &L^{2}_{\mathbf{K}_{\ast},\tau} \equiv \{ f \in L^{2}_{\mathbf{K}_{\ast}} : \mathcal R[f] = \tau f \}\\
    &L^{2}_{\mathbf{K}_{\ast},\overline{\tau}} \equiv \{ f \in L^{2}_{\mathbf{K}_{\ast} } : \mathcal R[f] = \overline{\tau} f \},
\end{align*}
to be used later on.

In \cite{FW1}, the following class of potentials, which yields such honeycomb lattice structures, has been introduced:
\begin{definition}
Let $V \in C^{\infty}(\mathbb{R}^{2};\R)$.  Then $V$ is called a \textbf{honeycomb lattice potential} 
if there exists an $\mathbf{y}_{0} \in \mathbb{R}^2 $ such that $\tilde{V} = V(\mathbf{y}-\mathbf{y}_{0})$ has the following properties : 
\begin{enumerate}
\item $\tilde{V}$ is $\Lambda$-periodic.
\item $\tilde{V}$ is inversion symmetric, i.e., $\tilde{V}(-\mathbf{y})=\tilde{V}(\mathbf{y})$.
\item $\tilde{V}$ is $\mathcal{R}-$invariant, i.e., $\tilde{V}(R^{*}\mathbf{y}) = \tilde{V}(\mathbf{y}).
$
\end{enumerate}
\end{definition}

\begin{remark} In physics experiments, honeycomb lattices typically are generated by the interference 
of three laser beams. For a concrete example of such a potential, see, e.g., \cite{AZ1,BPS}. 
\end{remark}

Next, we consider the Schr\"odinger operator
\begin{equation}\label{H}
H=-\Delta+V_{\rm per}({\bf y}), \quad {\bf y}\in \R^2,
\end{equation}
where $V_{\rm per}$ is assumed to be a smooth honeycomb lattice. Bloch-Floquet theory consequently asserts that the spectrum of $H$ is given by (see \cite{Wi}):
$$ \sigma(H) = \bigcup_{m \in \mathbb{N}} E_{m},
$$
where $E_{m}=\{\mu_{m}(\mathbf{k}) : \mathbf{k} \in Y^* \}$ is called the $m$-th energy band, or Bloch band. The eigenvalues $\mu_m({\bf k})$ are thereby obtained 
through the  ${\bf k}-$pseudo periodic eigenvalue problem \eqref{EVP}. The associated eigenfunctions $\Phi_m(\cdot; {\bf k})\in H^2_{\rm loc}(\R^2)$ are 
${\bf k}-$pseudo periodic, and usually referred to as Bloch waves. They form, for every fixed ${\bf k}\in Y^*$ 
a complete orthonormal basis on $L^2(Y)$ where from now on, we denote
$$
\langle f, g \rangle_{L^2(Y)}  = \frac{1}{|Y|} \int_{Y} \overline{f({\bf y})}  g({\bf y}) \, d{\bf y},
$$
the associated $L^2$-inner product.  We also note 
that for any $m \in \mathbb{N}$ there exists a closed subset 
$I \subset Y^*$ such that the functions $\mu_{m}^{2}(k)$ are real analytic functions for all $k \in Y^*/I$, and we have the following condition
$$
\mu_{m-1}(\mathbf{k}) < \mu_{m}(\mathbf{k}) < \mu_{m+1}(\mathbf{k}), \ \mathbf{k} \in Y^*/I.
$$
We call $E_{m}$ an isolated Bloch band if the condition above holds for all ${\bf k} \in Y^*$. Lastly, it is known that 
$$
|I| = |\{\mathbf{k} \in Y^*: \mu_{m}(\mathbf{k}) = \mu_{m+1}(\mathbf{k})   \}| = 0,
$$
and it is in this set of measure zero that we encounter what are called band crossings. 
At such band crossings one does not have differentiability in $\mathbf{k}$ of the eigenvalues and eigenfunctions.

\subsection{Dirac Points} An important feature of honeycomb lattice potentials is the presence of Dirac points. The following definition is taken from \cite{FW1}:

\begin{definition}\label{diracpoint}
Let $V$ be a smooth honeycomb lattice potential. Then a vertex $\textbf{K}={\bf K}_\ast \in Y^*$ is called a \textbf{Dirac point} if the following holds:
There exists $m_{1}\in \N$, a real number $\mu_{\ast}$, and strictly positive numbers, $\lambda$ and $\delta$, such that:
\begin{enumerate}
\item $ \mu_{\ast}$ is a degenerate eigenvalue of $H$ with associated $\textbf{K}_\ast$-pseudo-periodic eigenfunctions.
\item $\mbox{dim} \ \mbox{Nullspace}(H-\mu_{\ast})$ = 2.
\item  $\mbox{Nullspace}(H-\mu_{\ast}) = \mbox{span}\lbrace\Phi_{1},\Phi_{2} \rbrace$, where $\Phi_{1} \in L^{2}_{\textbf{K},\tau}$ and $\Phi_{2} \in L^{2}_{\textbf{K},\overline{\tau}}$.
\item There exists Lipschitz functions $\mu_{\pm}(\textbf{k})$,
\begin{equation*}
\mu_{m_{1}}(\textbf{k}) = \mu_{-}(\textbf{k}), \ \ \mu_{m_{1}+1}(\textbf{k}) = \mu_{+}(\textbf{k}), \ \ \mu_{m_{1}}(\textbf{K}_\ast) = \mu_{\ast},
\end{equation*}
and $E_{\pm}(\textbf{k})$, defined for $|\textbf{k} -\textbf{K}_\ast| < \delta$, such that 
\begin{align*}
\mu_{+}(\textbf{k}) - \mu_{\ast}&= + \lambda |\textbf{k} -\textbf{K}_\ast | (1+E_{+}(\textbf{k})) \\
\mu_{-}(\textbf{k}) - \mu_{\ast} &= - \lambda |\textbf{k} -\textbf{K}_\ast | (1+E_{-}(\textbf{k})),
\end{align*}
where $|E_{\pm}(\textbf{k})| < C|\textbf{k} -\textbf{K}_\ast|$ for some $C>0$.
\end{enumerate}
\end{definition}
For later purposes we also recall the following result from \cite{FW1} which is computed using the Fourier series expansion $\Phi_{1,2}$, spanning the two dimensional 
eigenspace associated to $\mu_\ast$. 
\begin{proposition}
\label{prop:comput}
Let $\zeta = (\zeta_{1},\zeta_{2}) \in \mathbb{C}^{2}$ some vector. Then it holds
\begin{align*}
&\langle \Phi_{n},  \zeta \cdot \nabla\Phi_{n}\rangle_{L^2(Y)} = 0 \ , \ n =1,2 \nonumber, \\
 2i & \langle\Phi_{1}, \zeta \cdot\nabla \Phi_{2}\rangle_{L^2(Y)} = \overline{2i \langle \Phi_{2}, \zeta \cdot\nabla \Phi_{1}\rangle}_{L^2(Y)} = -\overline{\lambda_{\#}}  (\zeta_{1} + i\zeta_{2}), \nonumber \\
 2i & \langle\Phi_{2}, \zeta \cdot\nabla \Phi_{1}\rangle_{L^2(Y)} = -\lambda_{\#} (\zeta_{1} - i\zeta_{2}),
\end{align*}
where $\lambda_{\#}\in \C$ is defined by
\begin{equation*}
\lambda_{\#} = 3 |Y|  \sum_{\mathbf{m} \in S}c(\mathbf{m})^{2} \begin{pmatrix}  1 \\ i \end{pmatrix} \cdot \mathbf{K}^\ast_{\mathbf{m}}.
\end{equation*}
Here $\lbrace c(\mathbf{m}) \rbrace_{\mathbf{m} \in S\subset \Z^2}$ denotes the sequence of $L^{2}_{\mathbf{K}_{\ast},\tau }$ Fourier coefficients of the normalized eigenstate $\Phi_{1}({\bf x})$ and 
$\mathbf{K}^\ast_{\mathbf{m}} = \mathbf{K}_\ast+m_1 {\bf k}_1 + m_2 {\bf k}_2$.
\end{proposition}
If $\lambda_{\#} \not =0$, then (4) in Definition \ref{diracpoint} above holds with $\lambda = |\lambda_{\#} |$. For the present paper, we shall thus 
make the standing assumption:

\begin{hyp}\label{ass}
$V_{\rm per}\in C^\infty(\R^2;\R)$ is a smooth honeycomb lattice potential, which admits Dirac points such that
$\lambda_{\#} \not =0$. 
\end{hyp}

It is proved in \cite{FW1}, that this assumption is generically satisfied.


\section{Multi-scale asymptotic expansions}\label{sec:multi}

\subsection{Formal derivation of the Dirac system}
In this section, we shall follow the ideas in \cite{CMS, GMS}, and perform a formal multi-scale expansion of the solution to \eqref{NLS} under the Assumption \ref{ass}. 
To this end, we seek a solution of the form
\[
\Psi^\eps(t,{\bf x})\Eq \eps 0 \Psi^\eps_N(t,{\bf x}):= \, e^{-i\lambda t / \varepsilon}\sum \limits_{n=0}^{N}\varepsilon^{n}u_{n}\Big( t, \mathbf{x},\frac{\mathbf{x}}{\varepsilon}\Big),\quad \lambda \in \R,
\]
where each $u_n(t, {\bf x}, \cdot)$ is supposed to be ${\bf k}-$pseudo periodic with respect to the fast variable ${\bf y} = \frac{\mathbf{x}}{\varepsilon}$. From now on, we denote the linear Hamiltonian by
\begin{equation}\label{ham}
 \displaystyle H^{\varepsilon} = -\varepsilon^{2}\Delta + V_{\rm per} \Big(\frac{\mathbf{x}}{\varepsilon}\Big)
 \end{equation}
and formally plug the ansatz $\Psi^\eps_N$ into \eqref{NLS}. This yields
\begin{equation*}
i\varepsilon\partial_{t}\Psi^{\varepsilon}_N - H^{\varepsilon}\Psi^{\varepsilon}_N - \varepsilon\kappa | \Psi^\eps_N|^2 \Psi^\eps_N= 
e^{-i\lambda (t / \varepsilon)}\sum \limits_{n=0}^{N}\varepsilon^{n} X_{n} + {\rho}(\Psi^{\varepsilon}_N),
\end{equation*}
where the remainder is 
\begin{equation}\label{remainder}
{\rho }\,(\Psi^{\varepsilon}_N) = e^{-i\lambda (t / \varepsilon)}\sum \limits_{n=N+1}^{3N+1}\varepsilon^{n} X_{n}.
\end{equation}
Introducing the operators
\[
L_{0} = \lambda -H, \quad L_{1} = i\partial_{t} + 2\nabla_{\mathbf{x}}\cdot \nabla_{\mathbf{y}}, \quad L_{2} = \Delta_{\mathbf{x}},
\]
with $H$ given by \eqref{H}, we find (after some lengthy calculations) that
\[
X_n = L_0 u_n + L_1 u_{n-1} + L_2 u_{n-2} - \kappa \sum_{ \mathclap{\substack{ j+k+l+1=n \\ 0\le j,k,l < n}}}u_{j}u_{k}\overline{u}_{l}.
\]

We can then proceed by solving $X_n =0$ for all $n\le N$, to obtain an approximate solution $\Psi_N^\eps$, which formally solves the NLS up to an
error of order $\mathcal O(\eps^{N+1})$. To this end, the leading order equation is $L_0 u_0 = 0$, which means
\[
Hu_0\equiv (- \Delta_{\bf y} + V_{\rm per} ({\bf y})) u_0 = \lambda u_0.
\]
In view of the assumption that $u_0$ is ${\bf k}-$pseudo periodic, this implies that $\lambda = \mu(\mathbf{k})$ is a Bloch eigenvalue. We shall from now on fix ${\bf k} = {\bf K}_\ast$ to be a Dirac point satisfying 
Assumption \ref{ass}, and denote the associated eigenvalue by $\lambda = \mu_\ast$. The leading order amplitude $u_0$ can thus be written as
\begin{equation}\label{u0}
u_{0}(t,\mathbf{x},\mathbf{y})=\sum \limits_{j=1}^{2}\alpha_{j}(t,\mathbf{x})\Phi_{j}(\mathbf{y};\textbf{K}_{\ast}),
\end{equation}
where $\Phi_{1,2}$ span the two-dimensional eigenspace of $\mu_\ast$, cf. Definition \ref{diracpoint}. 

To determine the leading order amplitudes  $\alpha_{1},\alpha_{2}$ we set $X_1 = 0$  to obtain
\begin{equation}\label{1storder}
L_{0}u_{1} = \kappa |u_{0}|^2 u_0 - L_{1}u_{0}.
\end{equation}
Explicitly, the right hand side reads
\begin{equation}
\kappa |u_{0}|^2 u_0 - L_{1}u_{0}
=  -i \sum \limits_{j,k,l=1}^{2} \Big( \Phi_{j}\partial_{t} \alpha_{j} - 2i\nabla_{\mathbf{x}}\alpha_{j}  \cdot \nabla_{\mathbf{y}}\Phi_{j}  +i \kappa \alpha_{j}\overline{\alpha}_{k}\alpha_{l}\Phi_{j}\overline{\Phi}_{k}\Phi_{l}\Big).
\end{equation}
By Fredholm's alternative, a necessary condition for the solvability of \eqref{1storder} is that the right hand side $\not \in \ker (L_{0})$. 
Denoting by $\mathbb P_\ast=\mathbb P_\ast^2 $ the $L^2(Y)-$projection on 
the spectral subspace corresponding to $\mu_\ast$, we consequently require
\[
\mathbb P_\ast ( \kappa |u_{0}|^2 u_0 - L_{1}u_{0})=0,
\]
or, more explicitly:
\[
\sum \limits_{j,k,l=1}^{2} \big \langle \Phi_{m}  , \Phi_{j}\partial_{t} \alpha_{j} -2i \nabla_{\mathbf{x}}\alpha_{j}  \cdot \nabla_{\mathbf{y}}\Phi_{j}  
+i \kappa \alpha_{j}\overline{\alpha}_{k}\alpha_{l}\Phi_{j}\overline{\Phi}_{k}\Phi_{l} \big \rangle_{L^2(Y)} = 0 , \ \mbox{for} \ m=1,2.
\]
Applying Proposition \ref{prop:comput} with $\zeta=\nabla_{\mathbf{x}}\alpha_{j} \in \mathbb{C}^{2}$ for $j=1,2$, respectively, we obtain the following system of equations
\begin{align*}
&  \partial_{t} \alpha_{1} + \overline{\lambda}_{\#} \big(\partial_{x_{1}} +i \partial_{x_{2}}\big)\alpha_{2} +i  \sum \limits_{j,k,l=1}^{2} \kappa_{(j,k,l,1)}\alpha_{j}\overline{\alpha}_{k}\alpha_{l}= 0 \nonumber \\ 
&  \partial_{t} \alpha_{2} + \lambda_{\#} \big(\partial_{x_{1}} -i \partial_{x_{2}}\big)\alpha_{1} +i \sum \limits_{j,k,l=1}^{2} \kappa_{(j,k,l,2)}\alpha_{j}\overline{\alpha}_{k}\alpha_{l} = 0,
\end{align*}
where $\lambda_{\#}\not =0$ is as above, and 
$$ \kappa_{(j,k,l,m)}= \kappa \big\langle \Phi_{m} , \Phi_{j}\overline{\Phi}_{k}\Phi_{l} \big\rangle_{L^2(Y)}.$$

The latter can now be evaluated, using symmetry considerations of the action 
of the operator $\mathcal{R}$, defined in \eqref{R}, when acting onto Bloch eigenfunctions. 
We recall that  $\Phi_{1} \in L^{2}_{\textbf{K},\tau}, \Phi_{2} \in L^{2}_{\textbf{K},\overline{\tau}}$ are eigenfunctions for $\mathcal{R}$ such that
\begin{align*}
\mathcal{R}[\Phi_{1}(\mathbf{y})]=\Phi_{1}(R^{\ast}\mathbf{y})=\tau\Phi_{1}(\mathbf{y}), \quad \mathcal{R}[\Phi_{2}(\mathbf{y})]=\overline{\Phi}_{1}(-R^{\ast}\mathbf{y})=\overline{\tau}\Phi_{2}(\mathbf{y}) ,
\end{align*}
where $\displaystyle \tau = e^{2\pi i/3}$. To compute the integrals  $\big\langle \Phi_{m} , \Phi_{j}\overline{\Phi}_{k}\Phi_{l} \big\rangle_{L^2(Y)}$ we use the change of coordinates $\mathbf{\mathbf{y}} =R^{\ast}\mathbf{w} $ and apply the relations above, to obtain:
\begin{align*}
\displaystyle \kappa_{(j,k,l,m)} &= \kappa \int_{R(Y)} \overline{\Phi}_{m}(R^{\ast}\mathbf{w})\Phi_{j}(R^{\ast}\mathbf{w})\overline{\Phi}_{k}(R^{\ast}\mathbf{w})\Phi_{l}(R^{\ast}\mathbf{w}) \;d\mathbf{w} \\
&= \kappa\int_{Y}\overline{\tau}_{m}\overline{\Phi}_{m}(\mathbf{w})\tau_{j} \Phi_{j}(\mathbf{w})\overline{\tau}_{k}\overline{\Phi}_{k}(\mathbf{w})\tau_{l}\Phi_{l}(\mathbf{w}) \;d\mathbf{w} 
= \tau_{j}\overline{\tau}_{k}\tau_{l}\overline{\tau}_{m}\kappa_{(j,k,l,m)},
\end{align*}
so that 
\[
\kappa_{(j,k,l,m)}\Big(1-C_{(j,k,l,m)}\Big) = 0,\quad C_{(j,k,l,m)} = \tau_{j}\overline{\tau}_{k}\tau_{l}\overline{\tau}_{m}.
\]
We consequently find that $\kappa_{(j,k,l,m)} $ {vanishes,} whenever 
$C_{(j,k,l,m)}  \neq 1$. A computation shows 
\begin{align*}
&C_{(j,j,j,m)}=\tau_{j}\overline{\tau}_{j}\tau_{j}\overline{\tau}_{m}= \tau_{j} \overline{\tau}_{m} \neq 1 \ \text{for} \ m \neq j,\\
&C_{(j,k,j,m)}= C_{(k,j,m,j)} =\tau^{2}_{j}\overline{\tau}_{k}\overline{\tau}_{m}= \tau_{k}\overline{\tau}_{k}\overline{\tau}_{m} = \overline{\tau}_{m} \neq 1 \ \text{for} \ m = 1,2,\\
&C_{(j,j,k,m)}=\tau_{j}\overline{\tau}_{j}\tau_{k}\overline{\tau}_{m}=\tau_{k}\overline{ \tau}_{m}  \neq 1 \ \text{for} \ k \neq m =j,
\end{align*}
and we consequently deduce that the only non-vanishing coefficients are
\begin{align*}
&b_{1}:= \kappa_{(1,1,1,1)} = \kappa_{(2,2,2,2)}  = \int_{Y} |\Phi_{1}(\mathbf{y})|^{4} \;d\mathbf{y} =\int_{Y} |\Phi_{2}(\mathbf{y})|^{4} \;d\mathbf{y},\\
&b_{2} := \kappa_{(1,1,2,2)}  = \kappa_{(1,2,2,1)} = \kappa_{(2,2,1,1)} = \kappa_{(2,1,1,2)} = \int_{Y} |\Phi_{1}(\mathbf{y})|^{2}|\Phi_{2}(\mathbf{y})|^{2} \;d\mathbf{y}.
\end{align*}
In summary, we find 
the nonlinear Dirac system as announced in \eqref{nldirac}:
\begin{equation*}
\left\{
\begin{split}
    &  \partial_{t} \alpha_{1} + \overline{\lambda}_{\#} \big(\partial_{x_{1}} + i\partial_{x_{2}}\big)\alpha_{2} +i \kappa\Big(b_{1} |\alpha_{1}|^{2} + 2b_{2}|\alpha_{2}|^{2}\Big)\alpha_{1} = 0, \\ 
    &  \partial_{t} \alpha_{2} + \lambda_{\#}  \big(\partial_{x_{1}} -i \partial_{x_{2}}\big)\alpha_{1} +i \kappa \Big( b_{1}|\alpha_{2}|^{2} + 2 b_{2}|\alpha_{1}|^{2}  \Big)\alpha_{2} =0.
    \end{split}
    \right.
\end{equation*}
Obviously, this system needs to be supplemented with initial data $\alpha_{1,0}$, $\alpha_{2,0}$, which we shall, for simplicity, assume to be in the Schwartz space $\mathcal S(\R^2)$.

\subsection{Higher order corrections} Assuming that the leading order amplitudes $\alpha_{1},\alpha_{2}$ satisfy the nonlinear Dirac system
allows us to proceed with our expansion and solve \eqref{1storder} for $u_1$. We obtain a 
unique solution in the form
\begin{equation}\label{u1}
u_{1}(t,\mathbf{x},\mathbf{y}) = \tilde{u}_{1}(t,\mathbf{x},\mathbf{y}) + u_{1}^{\perp}(t,\mathbf{x},\mathbf{y}),
\end{equation}
where $\tilde{u}_{1}\in (\ker L_{0})$ i.e. 
\begin{equation}\label{u1t}
\tilde{u}_{1}(t,\mathbf{x},\mathbf{y}) = \sum \limits_{j=1}^{2} \beta_{j}(t,\mathbf{x})\Phi_{j}(\mathbf{y},\mathbf{K}_{\ast}) 
\end{equation}
for some yet to be determined amplitudes $\beta_{1,2}$, and 
\begin{equation}\label{u1p}
u_{1}^{\perp}(t,\mathbf{x},\mathbf{y})= - L^{-1}_{0}\big( L_{1}u_{0}-\kappa |u_{0}|^2 u_0 \big)  \in (\ker L_{0})^{\perp}.
\end{equation}
Here we denote by $L_0^{-1}$ the partial inverse (or partial resolvent) of $L_0$, i.e.
\[
L^{-1}_{0}= (1- \mathbb P_\ast) ( \mu_\ast -H)^{-1} (1- \mathbb P_\ast).
\]
Note that at $t=0$, the function $u_{1}^{\perp}(0,\mathbf{x},\mathbf{y}) $ cannot be chosen, but rather has to be determined from the initial data $\alpha_{0,1}$, $\alpha_{0,2}$ according to the formula above.

Proceeding further, we need to determine the amplitudes $\beta_{1},\beta_{2}$ by setting $X_2 = 0$, i.e.
\[
L_0 u_2 = - L_{1}u_{1} - L_{2}u_{0} + \kappa\big(u^{2}_{0}\overline{u}_{1} + 2|u_{0}|^{2}u_{1}\big).
\]
By the same arguments as before, we obtain the following system of linear, inhomogeneous Dirac equations for $\beta_{1,2}$ as the corresponding solvability condition:
\begin{equation}\label{linD}
\left \{
\begin{aligned}
& \partial_{t} \beta_{1} + \overline{\lambda}_{\#} \big(\partial_{x_{1}}  +i \partial_{x_{2}}\big)\beta_{2} +i \sum \limits_{j,k,l=1}^{2}  \kappa_{(j,k,l,1)}\Big(\alpha_{j}\overline{\beta}_{k}\alpha_{l} + 2\beta_{j}\overline{\alpha}_{k}\alpha_{l}\Big) = i\Theta_{1} \\ 
&  \partial_{t} \beta_{2} + \lambda_{\#} \big(\partial_{x_{1}} -i \partial_{x_{2}}\big)\beta_{1} +i \sum \limits_{j,k,l=1}^{2}  \kappa_{(j,k,l,2)}\Big(\alpha_{j}\overline{\beta}_{k}\alpha_{l} + 2\beta_{j}\overline{\alpha}_{k}\alpha_{l}\Big)  =i\Theta_{2},
\end{aligned}
\right.
\end{equation}
where the right hand side source terms can be written as 
\begin{equation}\label{source}
\Theta_{n} = \Delta\alpha_{n} +\langle \Phi_{n}, L_{1}u_{1}^{\perp}\rangle_{L^2(Y)} - \kappa \langle \Phi_{n}, (u_{0}\overline{u}^{\perp}_{1} + 2\overline{u}_{0}u^{\perp}_{1})u_{0} \rangle_{L^2(Y)}.
\end{equation}
We have the freedom to choose vanishing initial data ${\beta_{1,2}}(0, {\bf x})=0$ for the system \eqref{linD}, which nevertheless will have a non-vanishing solution due to the source terms.
In summary, this allows us to write 
$$
u_{2}(t,\mathbf{x},\mathbf{y}) = \tilde{u}_{2}(t,\mathbf{x},\mathbf{y}) + u^{\perp}_{2}(t,\mathbf{x},\mathbf{y}),
$$
where as before $\tilde{u}_{2} \in (\ker L_{0})$ and $u^{\perp}_{2} \in (\ker L_{0})^{\perp}$. The latter is obtained by elliptic inversion on $(\ker L_{0})^{\perp}$, i.e.
\begin{equation}\label{u2p}
u_{2}(t,\mathbf{x},\mathbf{y}) = -L_{0}^{-1}\Big( L_{1}u_{1} + L_{2}u_{0} - \kappa\big(u^{2}_{0}\overline{u}_{1} + 2|u_{0}|^{2}u_{1}\big)  \Big).
\end{equation}
All higher order terms $u_n$ can then be determined analogously. However, since we are mainly interested in deriving the nonlinear Dirac system 
for the leading order amplitudes $\alpha_{1,2}$, we shall see that it is sufficient to stop our expansion at $N=2$. Note that in order to satisfy $X_2=0$, one does not need to determine the 
amplitudes within $\tilde u_2\in (\ker L_{0})$, which will simplify our treatment below.


\section{Mathematical framework for the approximate solution}\label{sec:math}

\subsection{Local well-posedness of the Dirac equations}\label{LWP} 
We aim to make the formal multi-scale expansion of the foregoing section mathematically rigorous. To do so, 
we shall in a first step construct a unique local in-time solution with sufficient smoothness 
for the nonlinear Dirac model \eqref{nldirac}. To this end, we will work in the 
Banach space $X = C( [0,T);H^{s}(\mathbb{R}^{2}))^{2}$ for $s>1$, endowed with the norm
$$\| \mathbf{u} \|_{X} = \sup_{0 \le t \le T}\| \mathbf{u}(t) \|_{H^{s}} \equiv \sup_{0 \le t \le T}\Big( \|u_{1}(t)\|^{2}_{H^{s}} + \|u_{2}(t)\|^{2}_{H^{s}}\Big)^{1/2},$$ 
where $\mathbf{u} = (u_{1},u_{2}) $. For the sake of notation, let us rewrite \eqref{nldirac} in a more compact form. Namely, let $\alp  = (\alpha_1, \alpha_2)$ and define
$\sigma^{\#}_{1}:=P_{\#}\sigma_{1} $, $\sigma^{\#}_{2}:=P_{\#}\overline{\sigma_{2} }$ where $\sigma_{1,2}$ are the Pauli spin matrices 
\[
 \sigma_{1} = \begin{pmatrix} 0 & 1 \\ 1 & 0 \end{pmatrix}, \ \sigma_{2} = \begin{pmatrix} 0 & -i \\ i & 0 \end{pmatrix},
 \]
and 
\[
P_{\#} = \begin{pmatrix} \overline{\lambda_{\#}} & 0 \\ 0 & \lambda_{\#} \end{pmatrix}.
\]
Then \eqref{nldirac} can be written as
\begin{equation}\label{NLD}
    i\partial_{t} \mbox{\boldmath $\alpha$} = -i(\mbox{\boldmath $\sigma$} \cdot \nabla) \mbox{\boldmath $\alpha$} + \kappa\mathbf{F}(\mbox{\boldmath $\alpha$}),\quad \alp_{\mid t=0} = \alp_0({\bf x}),
\end{equation} 
where we denote
\[
(\mbox{\boldmath $\sigma$} \cdot \nabla):=\sigma^{\#}_{1}\partial_{x_{1}} + \sigma^{\#}_{2}\partial_{x_{2}} 
\]
and $\mathbf{F}(\mbox{\boldmath $\alpha$})$ is the nonlinearity discussed in the following lemma.

\begin{lemma}\label{lem:nonl}
Consider the function 
\[ \displaystyle G(\alpha_{1},\overline{\alpha}_{1},\alpha_{2},\overline{\alpha}_{2}) = \frac{b_{1}}{2}\big( |\alpha_{1}|^{4} + |\alpha_{2}|^{4} \big) + 2b_{2}|\alpha_{1}|^{2}|\alpha_{2}|^{2},\]
with $ \mbox{\boldmath $\alpha$} = (\alpha_{1},\alpha_{2}) \in H^{s}(\mathbb{R}^{2})^{2}$. Then, the nonlinear vector field $\mathbf{F}(\mbox{\boldmath $\alpha$})= \begin{pmatrix} \partial_{\overline{\alpha}_{1}}G \\ \partial_{\overline{\alpha}_{2}}G \end{pmatrix}$ is a map from $H^{s}(\mathbb{R}^{2})^{2} \to H^{s}(\mathbb{R}^{2})^{2}$ for any $s>1$.
\end{lemma}
\begin{proof}
We note that for $s>1$, $H^{s}(\mathbb{R}^{2})\hookrightarrow L^\infty(\R^2)$ is a commutative algebra such that   
\begin{equation*}
   \forall f,g \in H^{s}(\mathbb{R}^{2}), \  \| fg \|_{H^{s}} \le C_{s} \| f \|_{H^{s}} \| g \|_{H^{s}}.
\end{equation*}
This directly yields   
\[
 \| \mathbf{F}(\mbox{\boldmath $\alpha$}) \|_{H^{s}} \le \big( b_{1} +2b_{2} \big)C_{s}^{2}\|\mbox{\boldmath $\alpha$}\|^{3}_{H^{s}} <\infty.
 \] \end{proof}
 
Using the Fourier transform, we also have that the linear time-evolution governed by the (strongly continuous) Dirac group
\[
(U(t) f)({\bf x})= {e^{t \sigma \cdot \nabla}} f({\bf x})\equiv \mathcal F^{-1}\Big( e^{-it \mbox{\boldmath $\sigma$} \cdot \mbox{\boldmath $\xi$}}\widehat f (\mbox{\boldmath $\xi$}) \Big)({\bf x}),
\]
satisfies $\| U(t) \mathbf{u} \|_{H^{s}}= \|  \mathbf{u} \|_{H^{s}}$ for all $s\in \R$ and all $t\in \R$. Together with the foregoing lemma, this can be used to prove the following local 
well-posedness result.

\begin{proposition}\label{prop:dirac}
For any $\mbox{\boldmath $\alpha$}_{0} \in \mathcal{S}(\mathbb{R}^{2})^{2}$ there exists a time $T>0$ and a unique maximal solution $ \mbox{\boldmath $\alpha$} 
\in C\big( [0,T);H^{s}(\mathbb{R}^{2})\big)\cap C^1 \big( (0,T); H^{s-1}(\mathbb{R}^{2})\big)^{2}$ for all $s>1$ to \eqref{NLD}. Moreover, it holds
\[
\| \alpha_1(t,\cdot) \|_{L^2} + \| \alpha_2 (t,\cdot)\|_{L^2}  = \ \text{\rm const.} \ \forall t\in [0,T).
\]
\end{proposition}

\begin{proof} This follows by a fixed point argument applied to Duhamel's formulation of \eqref{NLD}, i.e.
\begin{equation}
\alp(t) =    U(t){\alp}_{0} -i\kappa \int^{t}_{0}U(t-\tau)\mathbf{F}(\mathbf{u}(\tau))d\tau =: (\Phi  \alp) (t).
\end{equation}
Indeed, using Lemma \ref{lem:nonl}, one easily sees that for any $\alp_{0} \in \mathcal S(\R^2)^2$ such that 
$\|\alp_{0}\|_{H^{s}} \le R$, the functional $\Phi: X \to X $ is a contraction mapping on 
$K=\{ \mbox{\boldmath $\alpha$} \in X : \| \mbox{\boldmath $\alpha$} \|_{X}   \le 2R \}$, provided 
$T = \displaystyle \frac{1}{8(b_{1} + 2b_{2})R^{2}C^{2}_{s}}>0$. The asserted regularity in time then follows by differentiating the equation. 
Finally, the identity for mass conservation is obtained by multiplying the equations by $\overline \alpha_j$, integrating and taking the imaginary part.
\end{proof}
With this result in hand, we can also get local well-posedness for the inhomogeneous linear Dirac equation \eqref{linD} on the same time-interval, i.e., we have
\begin{equation}\label{beta}
\mbox{\boldmath $\beta$}  \in C\big( [0,T);H^{s-2}(\mathbb{R}^{2})\big)^{2},
\end{equation}
in view of the fact that the coefficients $\mbox{\boldmath $\alpha$}\in C([0,T); L^\infty(\R^2))^{2}$ and 
the source terms satisfy $\mbox{\boldmath $\theta$} \in C\big( [0,T);H^{s-2}(\mathbb{R}^{2})\big)^{2}$. The latter can be seen from \eqref{source} which 
involves the Laplacian of $u_0$ and thus we lose two derivatives.

\begin{remark}
The main obstacle to obtain global well-posedness of the nonlinear Dirac equation \eqref{NLD}, is the fact that the corresponding energy does not have a definite sign, i.e., 
\[
E(t) = {\rm Im} \left(\overline{\lambda_\#} \int_{\R^2} \alpha_2 (\partial_{x_1} + i \partial_{x_2} ) \overline \alpha_1 \, d{\bf x} \right) - \frac{\kappa}{4} \int_{\R^2} b_1 |\alp |^2 + 4 b_2 |\alpha_1|^2 |\alpha_2|^2 \, d{\bf x} .
\] 
So far, the existence of global in time solutions is thus restricted to small initial data cases, see, e.g., \cite{BeHe, EV, ES, MNO} and the references therein. The fact that 
\eqref{NLD} is also massless, is an additional complication, as the spectral subspaces for the corresponding free Dirac operator are no longer separated 
(a fact that requires considerable more care than the case with nonzero mass). 
\end{remark}


\subsection{Estimates on the approximate solution and the remainder} 
Formally, the approximate solution $\Psi_N^\eps$ derived in Section \ref{sec:multi} solves the NLS up to errors of order $\mathcal O(\eps^{N+1})$. 
To make this error estimate rigorous on time-scales of order $\mathcal O(1)$, 
we shall prove a nonlinear stability result in Section \ref{sec:stab} below. The latter will require us to work with an approximate solution of order $N>1$. We consequently 
need to work (at least) with $\Psi_2^\eps$. As was already remarked above, though, in order 
to solve the NLS up to reminders of $\mathcal O(\eps^2)$, one does not need to determine the 
highest order amplitudes within $\tilde u_2\in (\ker L_{0})$. We shall thus set them 
identically equal to zero and work with an approximate solution of the following form
\begin{equation}\label{approxsol}
\begin{aligned}
\Psi^{\varepsilon}_{\rm app}(t,\mathbf{x})= & \, e^{-i\mu_{\ast}t / \varepsilon}\Big ( u_0 + \eps u_1 +\eps^2 u_2\big)\Big( t, \mathbf{x},\frac{\mathbf{x}}{\varepsilon}\Big)\\
= & \, e^{-i\mu_{\ast}t / \varepsilon}\sum \limits_{j=1}^{2} \big(\alpha_{j} + \varepsilon\beta_{j}\big)(t,\mathbf{x})\Phi_{j}\Big(\frac{\mathbf{x}}{\varepsilon};\mathbf{K}_{\ast}\Big) +  \varepsilon^{j}u_{j}^{\perp}\Big(t,\mathbf{x},\frac{\mathbf{x}}{\varepsilon} \Big),
\end{aligned}
\end{equation}
where $\alpha_{1,2}\in C\big( [0,T);H^{s}(\mathbb{R}^{2})\big)$ and $\beta_{1,2}\in C\big( [0,T);H^{s-2}(\mathbb{R}^{2})\big)$ for $s>1$, 
are the leading and first order amplitudes guaranteed to exist by the results of the previous subsection. The term of order $\mathcal O(\eps^2)$ within this approximation 
is solely determined by elliptic inversion and thereby depends on the two lower order terms.

Since $\Psi^{\varepsilon}_{\rm app}$ involves the highly oscillatory Bloch eigenfuntions, we cannot expect to obtain uniform (in $\eps$) estimates in the usual Sobolev spaces $H^s(\R^2)$. 
We shall therefore work in the following $\eps$-scaled spaces, as used in \cite{CMS, GMS}.
\begin{definition}
Let $s \in \mathbb{Z}$, $0< \varepsilon \le 1 $, and $f \in H^{s}_{\varepsilon}(\mathbb{R}^{2})$ and define the $H^{s}_{\varepsilon}(\mathbb{R}^{2})-$norm by
$$
\| f \|^{2}_{H^{s}_{\varepsilon}}:= \sum \limits_{|\gamma|\le s} \|(\varepsilon \partial_{x})^{\gamma} f    \|^{2}_{L^2}.
$$ 
A family $f^{\varepsilon}$ is bounded in $H^{s}_{\varepsilon}(\mathbb{R}^{2})$ whenever 
$$
\sup_{0<\varepsilon\le 1} \| f^{\varepsilon} \|_{H^{s}_{\varepsilon}} < \infty.
$$
\end{definition}
We note  that in $H^{s}_{\varepsilon}(\mathbb{R}^{2})$ the following Gagliardo-Nirenberg inequality holds
\begin{equation}\label{GN}
\forall s>1 \ \exists C_{\infty}>0 : \ \| f \|_{L^{\infty}} \le  C_{\infty}\varepsilon^{-1}\| f \|_{ H^{s}_{\varepsilon} }, 
\end{equation}
where the ``bad'' factor $\varepsilon^{-1}$ is obtained by scaling. 

The following proposition then collects the necessary regularity estimates for our approximate solution and its corresponding remainder.
\begin{proposition}\label{prop:est}

Let $V_{\rm per}$ satisfy Assumption \ref{ass} and choose $S\in (4, \infty)$ such that $S>s+3$ for any $s>1$. Let $$\mbox{\boldmath $\alpha$}  \in 
C\Big([0,T),H^{S}(\mathbb{R}^{2})\Big)^2, \quad  \mbox{\boldmath $\beta$}  \in C\Big([0,T),H^{S-2}(\mathbb{R}^{2})\Big)^2$$ be the solutions to (non-)linear Dirac systems 
\eqref{NLD} and \eqref{linD}, respectively. 
Then, the approximate solution $ \Psi^{\varepsilon}_{\rm app}(t, \cdot)$, given in \eqref{approxsol}, satisfies the following estimates for all $t\in [0,T)$ and for any $|\gamma| < s-1$: 
$$
\|(\varepsilon\partial)^{\gamma}\Psi^\eps_{\rm app}(t,\cdot) \|_{L^{\infty}} \le C_{a}, \ \ \|\Psi^\eps_{\rm app}(t,\cdot) \|_{H^{s}_{\varepsilon}} \le C_{b}, \ \ \|\rho(\Psi^\eps_{\rm app})(t,\cdot)  \|_{H^{s}_{\varepsilon}} \le C_{r}\varepsilon^{3}.
$$
with $C_a, C_b, C_r>0$ independent of $\eps$.
\end{proposition}

\begin{proof}
To prove the estimates of the lemma, we need to first establish the regularity of $u_n(t)= u_{n}\big(t,\mathbf{x},{\mathbf{y}}\big)$ for $n=0,1,2$.  
We note first that Assumption \ref{ass} 
implies that the eigenfunctions $\Phi_{j}(\cdot \, ;\mathbf{K}^{\ast}) \in C^{\infty}(\overline{Y})$, $j=1,2$, cf. \cite{Wi}.  We shall then prove the following, 
preliminary estimate for $t \in [0,T)$
\begin{equation}\label{esti}
u_{n}(t) \in  H^{S-3}(\mathbb{R}^{2})\times C_{\mathbf{K}_{*}}^{\infty}(\overline{Y}), \quad n=0,1,2,
\end{equation}
where  $C_{\mathbf{K}_{*}}^{\infty}(\overline{Y})$ is the space of smooth $\mathbf{K}_{*}-$pseudo-periodic functions on $\overline{Y}$. To this end, we have 
for $t \in [0,T)$ that
$$ 
u_{0}(t) \in  H^{S}(\mathbb{R}^{2})\times C_{\mathbf{K}_{*}}^{\infty}(\overline{Y}),
$$
due to \eqref{u0} and Proposition \ref{prop:dirac}. Next, recall that $u_1$ is of the form \eqref{u1} with $\tilde u_1$ and $u_1^\perp$ 
given by \eqref{u1t} and \eqref{u1p}, respectively. In 
view of \eqref{u1t} and \eqref{beta}, we directly infer
$$  
\tilde{u}_{1}(t) \in H^{S-2}(\mathbb{R}^{2})\times C_{\mathbf{K}_{*}}^{\infty}(\overline{Y}).
$$ 
Moreover, since $L^{-1}_{0}:L^2(Y) \to H^2_{\rm per}(Y)$ and $L_{1}=i \partial_t + 2\nabla x \cdot \nabla_y$ 
it follows from \eqref{1storder} that 
$$ 
u_{1}^{\perp}(t) \in  H^{S-1}(\mathbb{R}^{2})\times C_{\mathbf{K}_{*}}^{\infty}(\overline{Y}),
$$
and thus 
$$
 u_{1}(t) \in H^{S-2}(\mathbb{R}^{2})\times C_{\mathbf{K}_{*}}^{\infty}(\overline{Y}),
$$
since $H^{s}(\R^2)\subset H^{s-1}(\R^2)$ for all $s>0$. Lastly, we recall that $u_2$ has the same 
type of structure with $u_2^\perp$ given by \eqref{u2p}. Similarly, as before it then follows that 
$$
u_{2}(t) \in  H^{S-3}(\mathbb{R}^{2}) \times C_{\mathbf{K}_{*}}^{\infty}(\overline{Y}).
$$
In summary, this yields \eqref{esti} which implies
$
u_{n}\Big(t,\cdot,\frac{\cdot}{\varepsilon}\Big) \in H_{\varepsilon}^{S-3}\big(\mathbb{R}^{2})$ for $n=0,1,2$. Hence it follows that for any $s>1$ 
\[
u_{n}\Big(t,\cdot,\frac{\cdot}{\varepsilon}\Big) \in H_{\varepsilon}^{s}\big(\mathbb{R}^{2}),\quad n=0,1,2,
\]
whenever $S-3>s$ which gives the condition stated above. Moreover it follows that 
\[
\|\Psi_{\rm app}(t,\cdot) \|_{H^{s}_{\varepsilon}} \le \sum \limits_{n=0}^{2}\varepsilon^{n}\|u_{n} \|_{H^{s}_{\varepsilon}} \le \sum_{n=0}^{2}\|u_{n} \|_{H^{S-3}_{\varepsilon}} = C_{b}.
\] 
Having established this, our next step is to prove the first inequality of the lemma. It suffices to show that there exists a constant $C_{0}>0$ such that 
\begin{equation}\label{est1}
\Big\|(\varepsilon\partial)^{\gamma}u_{0}\Big(t,\cdot,\frac{\cdot}{\varepsilon}\Big) \Big\|_{L^{\infty}} \le C_{0}
\end{equation}
holds for all $\varepsilon \in (0,1),|\gamma| < s-1$, and $t \in [0,T)$.
Since $u_{0}$ is of the form \eqref{u0}, we need only to show that
\[\Big\|(\varepsilon\partial)^{\gamma}\Big(\alpha_{1,2}(t,\cdot)\Phi_{1,2}\Big(\frac{\cdot}{\eps};\textbf{K}_{\ast}\Big)\Big) \Big\|_{L^{\infty}} \le C,
\]
where $C$ is a constant independent of epsilon.
By the Leibniz rule one has  
\[(\varepsilon\partial)^{\gamma}\Big(\alpha_{1,2}(t,\mathbf{x})\Phi_{1,2}\Big(\frac{\mathbf{x}}{\eps};\textbf{K}_{\ast}\Big)\Big) = \sum\limits_{\mathclap{ |\sigma| \le |\gamma|}} \begin{pmatrix} |\gamma| \\ \sigma\end{pmatrix} (\eps \partial)^{\sigma}\alpha_{1,2}(t,\mathbf{x}) \cdot (\eps \partial)^{\sigma-\gamma}\Phi_{1,2}\Big(\frac{\mathbf{x}}{\eps};\textbf{K}_{\ast}\Big),
\] 
and we can thus estimate 
\begin{align*} 
&\Big\|(\varepsilon\partial)^{\gamma}\Big(\alpha_{1,2}(t,\cdot)\Phi_{1,2}\Big(\frac{\cdot}{\eps};\textbf{K}_{\ast}\Big)\Big) \Big\|_{L^{\infty}} \\
&\le C_{1}\sum\limits_{\mathclap{ |\sigma| \le |\gamma|}} \Big\| (\eps \partial)^{\sigma}\alpha_{1,2}(t,\cdot) \Big\|_{L^{\infty}}  \Big\| (\eps \partial)^{\sigma-\gamma}\Phi_{1,2}\Big(\frac{\cdot}{\eps};\textbf{K}_{\ast}\Big) \Big\|_{L^{\infty}} \\
&\le C_{2}\Big\| \Phi_{1,2}\Big(\cdot;\textbf{K}_{\ast}\Big) \Big\|_{C_{\mathbf{K}_{\ast}}^{|\gamma|}} \, \sum\limits_{\mathclap{ |\sigma| \le |\gamma|}} \Big\|  \partial^{\sigma}\alpha_{1,2}(t,\cdot) \Big\|_{L^{\infty}}   \\
&\le C_{3} \Big\| \alpha_{1,2}(t,\cdot) \Big\|_{H^{2+|\gamma|}} \le C,
\end{align*}
where the second to last inequality follows by Sobolev embedding.  Hence \eqref{est1} follows by the triangle inequality. 
For $n=1,2$, one invokes \eqref{GN} directly to obtain
\[ \| (\eps\partial)^{\gamma}u_{n}\Big(t,\cdot,\frac{\cdot}{\eps}  \Big) \|_{L^{\infty}} \le C_{\infty}\eps^{-1}\| u_{n} \|_{H^{m +|\gamma|}_{\eps}} \le \eps^{-1} \| u_{n} \|_{H^{s}} \le \eps^{-1}C_{n},
\]
provided that $m>1$ and $m+|\gamma| \le s $, which implies $|\gamma| < s-1$. The desired inequality then follows, since
\[\|(\varepsilon\partial)^{\gamma}\Psi^\eps_{\rm app}(t,\cdot) \|_{L^{\infty}} \le C_{0} + \sum_{n=1}^{2}\eps^{n-1}\Big \| (\eps\partial)^{\gamma}u_{n}\Big(t,\cdot,\frac{\cdot}{\eps}  \Big) \Big\|_{L^{\infty}} \le C_{a},
\]
where $C_{a}$ is a constant independent of epsilon.
One notices that the last inequality follows from the fact that for $n=3,\dots, 7$: $X_{n} \in H_{\varepsilon}^{s}(\mathbb{R}^{2})$ and so 
\begin{equation*}
\|\rho(\Psi_{\rm app})(t,\cdot)  \|_{H^{s}_{\varepsilon}} \le \sum \limits_{n=3}^{7}\varepsilon^{n}\| X_{n} \|_{H^{s}_{\varepsilon}} \le C_{r}\varepsilon^{3}.
\end{equation*}
\end{proof}

With these estimates at hand, we can now prove the stability of our approximation. 


\section{Nonlinear stability of the approximation} \label{sec:stab}

Before we prove the nonlinear stability of our approximation, we recall that it was proved in \cite{GMS}, that the linear Schrödinger group 
\begin{equation}\label{group}
S^{\varepsilon}(t) :=e^{-iH^{\varepsilon}t/\varepsilon}
\end{equation}
generated by the periodic Hamiltonian $H^\eps$ defined in \eqref{ham} satisfies
\begin{equation}\label{propest}
\| S^{\varepsilon}(t)f \|_{H^{s}_{\varepsilon}} \le C_{s}\|f \|_{H^{s}_{\varepsilon}} \ \mbox{for all} \ t\in\mathbb{R} \ \mbox{and all} \ \varepsilon \in (0,1),
\end{equation}
where $C_s>0$ is independent of $\eps$. 
To this end, one uses the fact that $V_{\rm per}$ is assumed to be smooth and periodic (as guaranteed by Assumption \ref{ass}). Using this, it is straightforward to 
obtain the following basic well-posedness result for nonlinear Schr\"odinger equations.

\begin{lemma} \label{lem:exact} 
Let $V_{\rm per}$ satisfy Assumption \ref{ass} and $\Psi^{\varepsilon}_0 \in H_{\varepsilon}^{s}(\mathbb{R}^{2})$, for $s>1$. Then there exists a 
$T^\eps>0$ and a unique solution $ \Psi^{\varepsilon} \in C([0, T^{\varepsilon});H_{\varepsilon}^{s}(\mathbb{R}^{2}))$ to the initial value problem \eqref{NLS}. 
\end{lemma}
\begin{proof} In view of \eqref{propest}, the proof is a straightforward extension of the one given in, e.g., \cite[Proposition 3.8]{Tao} for the case without periodic potential. \end{proof}

\begin{remark} Note that this result does not preclude the possibility that $T^\eps\to 0_+$, as $\eps \to 0_+$. 
However, it will be a by-product of our main theorem below that this is not the case (at least not for the class of initial data considered in this work).\end{remark}

As a final preparatory step, we recall the following Moser-type lemma proved in \cite[Lemma 8.1.1]{Ra}.
\begin{lemma}\label{lem:moser}
Let $R>0$, $s \in [0,\infty)$, and $\mathcal{N}(z)=\kappa|z|^{2}z$ with $\kappa \in \mathbb{R}$.  Then there exists a $C_{s} = C_{s}(R,s,d,\kappa)$ such that if $f$ satisfies 
$$
\|(\varepsilon\partial)^{\gamma}f\|_{L^{\infty}} \le R \ \ \forall |\gamma|\le s, 
$$
and $g$ satisfies
$$
\|g\|_{L^{\infty}} \le R,
$$
then 
$$
\| \mathcal{N}(f+g) - \mathcal{N}(f) \|_{H^{s}_{\varepsilon}} \le C_{s}\| g \|_{H^{s}_{\varepsilon}}
$$
\end{lemma}
In contrast to other estimates (e.g., Schauder etc.), the above result has the advantage to result in a linear bound on the nonlinearity, a fact we shall use in the proof below. Indeed, we 
are now in the position to prove our main result.
\begin{theorem}
Let $V_{\rm per}$ satisfy Assumption \ref{ass}, let $\mbox{\boldmath $\alpha$} \in C([0,T),H^{S}(\mathbb{R}^{2}))^{2} $ be a solution to \eqref{NLS} and 
$\mbox{\boldmath $\beta$} \in  C([0,T),H^{S-2}(\mathbb{R}^{2}))^{2}$ be a solution to \eqref{linD} for some $S>s+3$ with $s >1$. 
Finally, assume that there is a $c>0$ such that the initial data $\Psi_0^\eps$ of \eqref{NLS} satisfies
$$
\Big \|\Psi^{\varepsilon}_0 - \big(u_{0} +\varepsilon u_{1} \big)\left(0,\cdot,  \frac{\cdot}{\eps}\right) \Big \|_{H^{s}_{\varepsilon}} \le c\varepsilon^{2}.
$$

Then, for any $T_{\ast} \in [0,T )$ there exists an $\varepsilon_{0}= \varepsilon_{0}(T_{\ast}) \in (0,1)$ and a constant 
$C>0$, such that for all $\varepsilon \in (0,\varepsilon_{0})$, the solution $\Psi^{\varepsilon} \in C([0, T_\ast);H_{\varepsilon}^{s}(\mathbb{R}^{2}))$ of \eqref{NLS} exists and moreover
$$
\sup_{0 \le t \le T_\ast} \Big \|\Psi^{\varepsilon}(t,\cdot) - e^{-i\mu_\ast t /\eps} \big(u_{0} +\varepsilon u_{1} \big)\left(t,\cdot,  \frac{\cdot}{\eps}\right) \Big \|_{H^{s}_{\varepsilon}} \le C\varepsilon^{2}.
$$
\end{theorem}

\begin{proof} Consider the difference in the exact and approximate solution 
$$
\varphi^{\varepsilon} = \Psi^{\varepsilon}- \Psi^\eps_{\rm app},
$$
then $\varphi^\eps(t, {\bf x})$ satisfies
\begin{align*}
i\partial_{t}\varphi^{\varepsilon} = \frac{1}{\varepsilon}H^{\varepsilon}\varphi^{\varepsilon} + \Big(\mathcal{N}(\Psi^\eps_{\rm app} + \varphi^{\varepsilon}) - 
\mathcal{N}(\Psi^\eps_{\rm app})\Big) - \frac{1}{\varepsilon}\rho(\Psi^\eps_{\rm app}) , \quad 
\varphi^{\varepsilon}_{\mid t=0} =\varphi_0^\eps,
\end{align*}
where $H^{\varepsilon}$ is defined in \eqref{ham}, the nonlinearity is $\mathcal{N}(z)=\kappa|z|^{2}z$, and $\rho(\Psi^\eps_{\rm app})$ is the remainder obtained in \eqref{remainder} for $N=2$. We denote 
$w^{\varepsilon}(t) := \|\varphi^{\varepsilon}(t) \|_{H^{s}_{\varepsilon}} $ and first note by assumption that 
$$
w^{\varepsilon}(0)\equiv \| \varphi^\eps_0\|_{H^s_\eps} 
\le \Big \|\Psi^{\varepsilon}(0,\cdot) - \big(u_{0} +\varepsilon u_{1} \big)\left(0,\cdot,  \frac{\cdot}{\eps}\right)\Big \|_{H^{s}_{\varepsilon}} + \varepsilon^{2}\| u^{\perp}_{2}(0,\cdot) \|_{H^{s}_{\varepsilon}} \le \tilde{c}\varepsilon^{2},
$$ 
for some constant $\tilde{c} \in \mathbb{R}$. We shall 
prove that there exists $\tilde{C}>0$, $\varepsilon_{0} \in (0,1)$, such that for all $\varepsilon \in (0,\varepsilon_{0}] $ that $w^{\varepsilon}(t) \le \tilde{C}\varepsilon^{2}$ for $t\in [0,T_\ast]$.

To this end, we rewrite the equation above using Duhamel's principle, i.e.
\begin{align*}
\varphi^\eps (t) = & \, S^{\varepsilon}(t)\varphi^{\varepsilon}_{0}  
 - i\int \limits_{0}^{t}S^{\varepsilon}(t-\tau)\big(\mathcal{N}\big(\Psi^\eps_{\rm app}(\tau) + \varphi^{\varepsilon}(\tau)) - \mathcal{N}(\Psi^\eps_{\rm app}(\tau)\big)\big)\;d\tau \\
 & \, + \frac{i}{\varepsilon}\int \limits_{0}^{t}S^{\varepsilon}(t-\tau)\rho(\Psi^\eps_{\rm app}(\tau))\;d\tau,
\end{align*}
where $S^\eps(t)$ is the Schr\"odinger group \eqref{group}. 
Now, using the propagation estimate \eqref{propest}, together with the estimate on the remainder stated in Proposition \ref{prop:est}, we obtain 
$$
w^{\varepsilon}(t) \le C_{l}(\tilde{c} + C_{r}T_\ast)\varepsilon^{2} + C_{l}\int \limits_{0}^{t}\big\|\mathcal{N}\big(\Psi^\eps_{\rm app}(\tau) + 
\varphi^{\varepsilon}(\tau)\big) - \mathcal{N}\big(\Psi^\eps_{\rm app}(\tau)\big)\big\|_{H^{s}_{\varepsilon}}\;d\tau.
$$
for all $t\le T_\ast$.

Set $\tilde{C}:=C_{l}(\tilde{c} + C_{r}T_\ast)e^{C_{l}C_{s}T_\ast}$ and choose $M>\max\{\tilde{c},\tilde{C}\}$ such that $M\eps^{2} > \tilde{c}\varepsilon^{2} \ge w^{\varepsilon}(0)$. 
Since $w^{\varepsilon}(t)$ is continuous in time, there exists, for every $\varepsilon \in (0,1)$, a positive time $t_{M}^{\varepsilon}>0$, such that $w^{\varepsilon}(t) \le \varepsilon^{2} M$ for $t \le t_{M}^{\varepsilon}$.
The Gagliardo-Nirenberg inequality \eqref{GN} then yields for $s>1$
$$
\|\varphi^{\varepsilon}(t)\|_{L^{\infty}} \le \varepsilon^{-1}C_{\infty}\|  \varphi^{\varepsilon}(t) \|_{H^{s}_{\varepsilon}} = \eps^{-1} C_{\infty} w^{\varepsilon}(t) \le \eps C_{\infty}M
$$
for $t \le t_{M}^{\varepsilon}$. Hence there exists an $\varepsilon_{0} \in (0,1)$, such that 
$$
\|\varphi^{\varepsilon}(t)\|_{L^{\infty}} \le C_{a},
$$
for $\varepsilon \in (0,\varepsilon_{0}]$ and $t\le t_{M}^{\varepsilon}$. By Proposition \ref{prop:est} we also have 
$$ \|(\varepsilon\partial)^{\gamma}\Psi_{\rm app}(t,\cdot) \|_{L^{\infty}} \le C_{a} $$ 
for $|\gamma|\le s-1, \varepsilon \in (0,1),$ and $t<T$. Thus we are in a position to apply the Moser-type Lemma \ref{lem:moser} 
with $R=C_{a}$ to obtain
$$
w^{\varepsilon}(t) \le C_{l}(c + C_{r}T_\ast)\varepsilon^{2} + C_{l}C_{s}\int \limits_{0}^{t}w^{\varepsilon}(\tau)\;d\tau, \ \mbox{for} \ \varepsilon \in (0,\varepsilon_{0}] \ \mbox{and} \ t \le t_{M}^{\varepsilon}.
$$
Gronwall's lemma then yields
$$
w^{\varepsilon}(t) \le C_{l}(\tilde{c} + C_{r}T_\ast)\varepsilon^{2}e^{C_{l}C_{s}t} \le \tilde{C}\varepsilon^{2} < M \ \mbox{for all} \ \varepsilon \in (0,\varepsilon_{0}] \ \mbox{and} \ t \le T_\ast.
$$
Hence for any choice of $T_\ast$, by continuity, we may extend further in time in that $t_{M}^{\varepsilon} \geq T_\ast$, and so we have proved that for any $T_\ast < T$:
$$
\sup_{0 \le t \le T_\ast} w^{\varepsilon}(t) \equiv \sup_{0 \le t \le T_\ast}\|\Psi^{\varepsilon}(t,\cdot) - \Psi^\eps_{\rm app}(t,\cdot) \|_{H^{s}_{\varepsilon}}  \le \tilde{C}\varepsilon^{2}.
$$
Another triangle inequality, then yields the desired result. 
\end{proof}

This theorem implies the approximation result announced in \eqref{approx}, uniformly on finite time-intervals bounded by the local existence time of \eqref{NLD}, 
provided the initial data is sufficiently well prepared, i.e. up to errors of order $\mathcal O(\eps^2)$.

\begin{remark}\label{rem} 
Unfortunately, our proof requires us to work with an approximate solution of order $\mathcal O(\eps^2)$, in which we need to control also the first order corrector $\propto u_1$. 
In turn, this requires the initial data to be somewhat better than one would like it to be.
The reason for this is the scaling factor $\eps^{-1}$ appearing in the 
Gagliardo-Nirenberg inequality \eqref{GN}. 
It is conceivable that this is merely a technical issue which 
can be overcome by the use of a different functional framework. For example, the authors in \cite{DoHe, DoUe} work in a Wiener-type 
space for the Bloch-transformed solution and its approximation.
\end{remark}


\section{The case of Hartree nonlinearities}\label{sec:hartree}

In this section, we shall briefly discuss the case of NLS with Hartree nonlinearity, i.e., instead of \eqref{NLS}, we consider
\begin{equation}\label{hartree}
i\varepsilon\partial_{t}\Psi^{\varepsilon}   = - \varepsilon^{2}\Delta \Psi^{\varepsilon}  + V_{\rm per}\Big(\frac{\mathbf{x}}{\varepsilon}\Big)\Psi^{\varepsilon} + 
\varepsilon \kappa\left(\frac{1}{|\cdot|} \ast |\Psi^{\varepsilon}|^{2}\right) \Psi^{\varepsilon}, \quad \Psi^\eps_{\mid t=0} = \Psi^\eps_0({\bf x}).
\end{equation}
This model describes the (semi-classical) dynamics of electrons inside a graphene layer, under the influence of a self-consistent electric field. 

We again seek an asymptotic expansion of the form
\[
\Psi^{\varepsilon}(t,\mathbf{x})  \Eq \eps 0 e^{-it\mu_{\ast} / \varepsilon}\sum \limits_{n=0}^{2}\varepsilon^{n}u_{n}\Big( t, \mathbf{x},\frac{\mathbf{x}}{\varepsilon}\Big),
\]
where as before the $u_n$ are assumed to be ${\bf k}-$pseudo periodic with respect to the fast variable. By the same procedure as in Section \ref{sec:multi}, we obtain that
\[
u_{0}(t,\mathbf{x},\mathbf{y})=\sum \limits_{j=1}^{2}\alpha_{j}(t,\mathbf{x})\Phi_{j}(\mathbf{y};\textbf{K}_{\ast}).
\]
Plugging this into the Hartree nonlinearity, yields the following nonlinear potential
$$
V^{\varepsilon} (t, {\bf x})= \frac{1}{|\cdot|}\ast  \sum \limits_{j,k=1}^{2} \big(\alpha_{j}\overline{\alpha}_{k}\Phi_{j}\overline{\Phi}_{k} \big)\Big(t,\cdot,\frac{\mathbf{\cdot}}{\varepsilon}\Big),
$$
which unfortunately does not directly exhibit the required two-scale structure. However, we shall prove the following averaging result.

\begin{lemma}
Let $\alpha_{j}(t,\mathbf{x})\in H^s(\R^2)$, for $s>1$, then
$$ \lim_{ \varepsilon \to 0 }V^{\varepsilon}(t,\mathbf{x}) = \Big(\frac{1}{|\cdot|}\ast \big( |\alpha_{1}|^{2} + |\alpha_{2}\big|^{2}\big) \Big)(t,\mathbf{x}).
$$
\end{lemma}

\begin{proof}
We recall that Bloch eigenfunctions concentrated at a Dirac point have a Fourier series expansion of the following form, cf. \cite{FW1}:
\begin{align*}
&\Phi_{1}(\mathbf{y},
\mathbf{K}^{\ast}) = \sum \limits_{\mathbf{m} \in \mathbb{Z}^{2}}c(\mathbf{m})e^{i\mathbf{K}^{\ast}_{\mathbf{m}}\cdot \mathbf{x}}, \\
&\Phi_{2}(\mathbf{x},
\mathbf{K}^{\ast}) = \overline{\Phi}_{1}(-\mathbf{x},
\mathbf{K}^{\ast}) = \sum \limits_{\mathbf{m} \in \mathbb{Z}^{2}}\overline{c}(\mathbf{m})e^{i\mathbf{K}^{\ast}_{\mathbf{m}}\cdot \mathbf{x}},
\end{align*}
where $$\displaystyle \mathbf{K}^{\ast}_{\mathbf{m}} = \mathbf{K}_{\ast} + m_{1}\mathbf{k}_{1} + m_{2}\mathbf{k}_{2} = \mathbf{K}_{\ast} + \mathbf{k}_{\mathbf{m}}.$$ 
By orthogonality and Parseval's identity one finds the following relations to be used below
\begin{align*}
&\sum \limits_{\mathbf{m}\in \mathbb{Z}^{2}} |c(\mathbf{m})|^{2} = \int_{Y}|\Phi_{1}(\mathbf{x})|^{2}\;d\mathbf{x} = 1, \\
& \sum \limits_{\mathbf{m}\in \mathbb{Z}^{2}} c(\mathbf{m})^{2} = \langle c(\mathbf{m}), \overline{c}(\mathbf{m}) \rangle_{\ell_{2}(\mathbb{Z}^{2})} = \int_{Y}\Phi_{1}(\mathbf{x})\overline{\Phi}_{2}(\mathbf{x})\;d\mathbf{x} = 0, \\
&\sum \limits_{\mathbf{m}\in \mathbb{Z}^{2}} \overline{c}(\mathbf{m})^{2} = \langle \overline{c}(\mathbf{m}), c(\mathbf{m}) \rangle_{\ell_{2}(\mathbb{Z}^{2})} = \int_{Y}\Phi_{2}(\mathbf{x})\overline{\Phi}_{1}(\mathbf{x})\;d\mathbf{x} = 0.
\end{align*}
Next, we decompose
\begin{align*} 
V^{\varepsilon} =  V_{1}^{\varepsilon} + V_{2}^{\varepsilon}\equiv
 \frac{1}{|\cdot|}\ast \sum \limits_{j=1}^{2} |\alpha_{j}|^2 |\Phi_{j}|^2 + \frac{1}{|\cdot|}\ast \sum \limits_{j\not =k }  \alpha_{j}\overline{\alpha}_{k}\Phi_{j}\overline{\Phi}_{k} .
\end{align*}
Explicitly, we find
\begin{align*}
V_{1}^{\varepsilon}(t,\mathbf{x}) &= \int_{\mathbb{R}^{2}} \Big(\big|\alpha_{1}(t,\mbox{\boldmath $\eta$})\big|^{2}\big|\Phi_{1}\Big(\frac{\mbox{\boldmath $\eta$}}{\varepsilon}\Big)\big|^{2} + \big|\alpha_{2}(t,\mbox{\boldmath $\eta$})\big|^{2}\big|\Phi_{2}\Big(\frac{\mbox{\boldmath $\eta$}}{\varepsilon}\Big)\big|^{2}\Big)\frac{d\mbox{\boldmath $\eta$}}{\big|\mathbf{x} - \mbox{\boldmath $\eta$}\big|} \\
&= \sum \limits_{\mathbf{m},\mathbf{m}^{\prime}\in \mathbb{Z}^{2}} c(\mathbf{m})\overline{c}(\mathbf{m}^{\prime})\int_{\mathbb{R}^{2}} \Big(\big|\alpha_{1}(t,\mbox{\boldmath $\eta$})\big|^{2} + \big|\alpha_{2}(t,\mbox{\boldmath $\eta$})\big|^{2}\Big)e^{i(\mathbf{k}_{\mathbf{m}-\mathbf{m}^{\prime}}\cdot \mbox{\boldmath $\eta$})/\varepsilon}\frac{d\mbox{\boldmath $\eta$}}{\big|\mathbf{x} - \mbox{\boldmath $\eta$}\big|},
\end{align*}
and similarly, 
\begin{align*}
&\, V_{2}^{\varepsilon}(t,\mathbf{x}) =\\
&= \int_{\mathbb{R}^{2}} \Big(\alpha_{1}(t,\mbox{\boldmath $\eta$})\overline{\alpha}_{2}(t,\mbox{\boldmath $\eta$})\Phi_{1}\Big(\frac{\mbox{\boldmath $\eta$}}{\varepsilon}\Big)\overline{\Phi}_{2}\Big(\frac{\mbox{\boldmath $\eta$}}{\varepsilon}\Big) + \overline{\alpha}_{1}(t,\mbox{\boldmath $\eta$})\alpha_{2}(t,\mbox{\boldmath $\eta$})\overline{\Phi}_{1}\Big(\frac{\mbox{\boldmath $\eta$}}{\varepsilon}\Big)\Phi_{2}\Big(\frac{\mbox{\boldmath $\eta$}}{\varepsilon}\Big)\Big)\frac{d\mbox{\boldmath $\eta$}}{\big|\mathbf{x} - \mbox{\boldmath $\eta$}\big|} \\
&=\sum \limits_{\mathbf{m},\mathbf{m}^{\prime}\in \mathbb{Z}^{2}} c(\mathbf{m})c(\mathbf{m}^{\prime})\int_{\mathbb{R}^{2}} \alpha_{1}(t,\mbox{\boldmath $\eta$})\overline{\alpha}_{2}(t,\mbox{\boldmath $\eta$})e^{i(\mathbf{k}_{\mathbf{m}-\mathbf{m}^{\prime}}\cdot \mbox{\boldmath $\eta$})/\varepsilon} \frac{d\mbox{\boldmath $\eta$}}{\big|\mathbf{x} - \mbox{\boldmath $\eta$}\big|}  \\
& \quad  + \sum \limits_{\mathbf{m},\mathbf{m}^{\prime}\in \mathbb{Z}^{2}} \overline{c}(\mathbf{m})\overline{c}(\mathbf{m}^{\prime})\int_{\mathbb{R}^{2}} \overline{\alpha}_{1}(t,\mbox{\boldmath $\eta$})\alpha_{2}(t,\mbox{\boldmath $\eta$})e^{-i(\mathbf{k}_{\mathbf{m}-\mathbf{m}^{\prime}}\cdot \mbox{\boldmath $\eta$})/\varepsilon} \frac{d\mbox{\boldmath $\eta$}}{\big|\mathbf{x} - \mbox{\boldmath $\eta$}\big|}.
\end{align*}
Now, since the kernel $\frac{1}{|\cdot|}$ in two spatial dimensions is integrable at the origin and since $\alpha_{1,2}\in L^2{(\R^2)} \cap L^\infty(\R^2)$, the Riemann-Lebesgue lemma 
implies that, as $\varepsilon \to 0$, all $\eps-$oscillatory terms vanish, i.e. all terms for which $\mathbf{m}\not =\mathbf{m}^{\prime}$. In view of the above identities for the Fourier coefficients we thus find
\begin{align*}
\lim_{\eps \to 0} V_{1}^{\varepsilon}(t,\mathbf{x}) &= \sum \limits_{\mathbf{m}\in \mathbb{Z}^{2}} |c(\mathbf{m})|^{2}\int_{\mathbb{R}^{2}} \Big(\big|\alpha_{1}(t,\mbox{\boldmath $\eta$})\big|^{2} + \big|\alpha_{2}(t,\mbox{\boldmath $\eta$})\big|^{2}\Big)\frac{d\mbox{\boldmath $\eta$}}{\big|\mathbf{x} - \mbox{\boldmath $\eta$}\big|}\\
&=\bigg(\frac{1}{|\cdot|}\ast \Big( \big|\alpha_{1}\big|^{2} + \big|\alpha_{2}\big|^{2}\Big) \bigg)(t,\mathbf{x}),
\end{align*}
whereas 
\begin{align*}
\lim_{\eps \to 0} V_{2}^{\varepsilon}(t,\mathbf{x}) = & \, \sum \limits_{\mathbf{m}\in \mathbb{Z}^{2}} c(\mathbf{m})^{2}\int_{\mathbb{R}^{2}} \alpha_{1}(t,\mbox{\boldmath $\eta$})\overline{\alpha}_{2}(t,\mbox{\boldmath $\eta$})\frac{d\mbox{\boldmath $\eta$}}{\big|\mathbf{x} - \mbox{\boldmath $\eta$}\big|} \\
& \, + \sum \limits_{\mathbf{m}\in \mathbb{Z}^{2}} \overline{c}(\mathbf{m})^{2}\int_{\mathbb{R}^{2}} \overline{\alpha}_{1}(t,\mbox{\boldmath $\eta$})\alpha_{2}(t,\mbox{\boldmath $\eta$})\frac{d\mbox{\boldmath $\eta$}}{\big|\mathbf{x} - \mbox{\boldmath $\eta$}\big|}=0,
\end{align*}
since both sums vanish individually.
\end{proof}

We therefore expect that as $\eps \to 0$, the dynamics of WKB waves spectrally localized around ${\bf K_\ast}$ are governed by the following Dirac-Hartree system:
\begin{equation*}\label{dirachartreee}
\left \{
\begin{aligned} 
& \partial_{t} \alpha_{1} + \overline{\lambda}_{\#} \big(\partial_{x_{1}} + i\partial_{x_{2}}\big)\alpha_{2} = -i\kappa\bigg(\frac{1}{|\cdot|}\ast \Big( \big|\alpha_{1}\big|^{2} + \big|\alpha_{2}\big|^{2}\Big) \bigg)\alpha_{1} , \\ 
    &  \partial_{t} \alpha_{2} + \lambda_{\#}  \big(\partial_{x_{1}} -i \partial_{x_{2}}\big)\alpha_{1} = -i \kappa \bigg(\frac{1}{|\cdot|}\ast \Big( \big|\alpha_{1}\big|^{2} + \big|\alpha_{2}\big|^{2}\Big) \bigg) \alpha_{2} ,
   \end{aligned}
    \right.
\end{equation*}
The Cauchy problem for this system has been rigorously studied in \cite{HaMe} as an ad-hoc model for electrons in graphene (see also \cite{HeTe, NaTs} for related results). 
Our analysis above indicates, that this is indeed the correct model and we 
believe that a fully rigorous proof can be achieved along the same lines as in the case of a cubic nonlinearity. However, a rigorous averaging argument for the required second order approximate solution would 
require considerably more work, a direction we do not want to pursue here. 

\begin{remark}
Let us finally note that while the Hartree nonlinearity is equivalent to a coupling with a Poisson equation $-\Delta V = | \Psi^\eps |^2$ in three spatial dimensions, this is no longer the case in 2D. 
If one were to pursue the coupled system instead of \eqref{hartree}, the effective model obtained would be a Dirac-Poisson system as studied in \cite{ChGl}.
\end{remark}


\bibliographystyle{amsplain}

\end{document}